\documentclass[12pt]{article}
\usepackage[utf8]{inputenc}

\usepackage{amsmath}
\usepackage{amssymb}
\usepackage{amsthm}
\usepackage{authblk}

\usepackage{algpseudocode}
\usepackage{algorithm}
\usepackage{tikz}
\usetikzlibrary{positioning}

\usepackage{fullpage}

\usepackage{xcolor}

\usepackage{hyperref}

\title{On Czerwinski's ``$\p\neq\np$ relative to a $\p$-complete oracle''\thanks{Supported in part by NSF grant 
		CCF-2006496.}}

\author{Michael~C.~Chavrimootoo} 
\author{Tran~Duy~Anh~Le}
\author{Michael~P.~Reidy}
\author{Eliot~J.~Smith}
\affil{Department of Computer Science\\University of Rochester\\Rochester, NY 14627, USA}

\newcommand{\naturalnumber}{\ensuremath{{\mathbb{N}}}}
\newcommand{\naturalnumberpositive}{\ensuremath{{\mathbb{N}^+}}}

\newcommand{\p}{\ensuremath{{\rm P}}}
\newcommand{\np}{\ensuremath{{\rm NP}}}
\newcommand{\dec}{\ensuremath{{\rm \Sigma_0^0}}}
\newcommand{\re}{\ensuremath{{\rm \Sigma_1^0}}}

\newcommand{\pspace}{\ensuremath{{\rm PSPACE}}}

\newcommand{\pair}[1]{\langle #1 \rangle}

\newcommand{\up}{\ensuremath{{\rm U}_{\rm P}}}
\newcommand{\unp}{\ensuremath{{\rm U}_{\rm NP}}}
\newcommand{\detp}{\ensuremath{{\rm D}_{\rm P}}}
\newcommand{\detnp}{\ensuremath{{\rm D}_{\rm NP}}}

\newcommand{\hp}{\ensuremath{{\rm HP}}}

\newcommand{\calo}{\mathcal{O}}
\newcommand{\calc}{\mathcal{C}}
\newcommand{\cald}{\mathcal{D}}

\newcommand{\calx}{\mathcal{X}}

\newcommand{\turing}{\ensuremath{\leq_{T}}}
\newcommand{\boundedturing}[1]{\ensuremath{\leq_{#1\hbox{-}T}}}
\newcommand{\manyone}{\ensuremath{\leq_m}}

\newtheorem{theorem}{Theorem}
\newtheorem{proposition}[theorem]{Proposition}
\newtheorem{lemma}[theorem]{Lemma}

\date{December 7, 2023}

\begin{document}\sloppy

\maketitle

\begin{abstract}
 In this paper, we take a closer look at Czerwinski's ``$\p\neq\np$ relative to a $\p$-complete oracle''~\cite{cze:t:p-neq-np-p-complete}. 
 There are (uncountably) infinitely-many relativized worlds where $\p$ and $\np$ differ, and it is well-known that for any $\p$-complete problem $A$, $\p^A \neq \np^A \iff \p\neq\np$. 
 The paper defines two sets $\detp$ and $\detnp$ and builds the purported proof of their main theorem on the claim that an oracle Turing machine with $\detnp$ as its oracle and that accepts $\detp$ must make $\Theta(2^n)$ queries to the oracle.
 We invalidate the latter by proving that there is an oracle Turing machine with $\detnp$ as its oracle that accepts $\detp$ and yet only makes one query to the oracle. We thus conclude that Czerwinski's paper~\cite{cze:t:p-neq-np-p-complete} fails to establish that $\p \neq \np$.
\end{abstract}

\section{Introduction}

This critique looks at Czerwinski's paper titled ``$\p \neq\np$ relative to a $\p$-complete oracle''~\cite{cze:t:p-neq-np-p-complete}, which asserts that while it does not resolve the $\p$ versus $\np$ problem, its method may bring insights into the resolution of the problem. Let us first discuss the importance and relevance of such a result.

Complexity theory is rife with hard, open problems that we do not know how to resolve using current methods. The most notable example is the $\p$ versus $\np$ problem that has been under consideration for several decades. So what should one do when faced with such obstacles? When we hope (and are unable to provide a proof) that a certain ``complexity statement'' $\calx$ (e.g., ``$\p = \np$'') does not hold, it is common to construct a ``relativized world'' in which that statement does not hold. The direct implication of such a result is that a proof of statement $\calx$ cannot relativize, i.e., cannot hold in every relativized world. The common view at that point is that $\calx$ is either false, or a proof of $\calx$ requires nonrelativizing proof techniques, thus making finding a proof a much harder challenge.

Of the many relativized worlds that have been given in the complexity theory literature, there are two important (classic) worlds that concern the $\p$ versus $\np$ problem. Baker, Gill, and Solovay~\cite{bak-gil-sol:j:rel} showed that there is an oracle $A$ where $\p^A=\np^A$ (in their paper, $A$ is some $\pspace$-complete set), and there is an oracle $B$ where $\p^B \neq \np^B$ (which they construct by diagonalization). Since we already know of an oracle relative to which $\p\neq\np$, why critique a paper that provides yet another oracle relative to which $\p\neq\np$? Well, it is a standard observation taught in introductory courses that for any $L \in \p$, $\p^L = \p$ and also $\np^L=\np$. So Czerwinski's paper~\cite{cze:t:p-neq-np-p-complete}, which claims to give a $\p$-complete oracle $X$ such that $\p^X \neq \np^X$, is one that is in fact---and it seems unbeknownst to the author as they claim it is unknown if $\p^\p=\p$ and it is unknown if $\np^\p=\np$---claiming to prove that $\p \neq \np$.

In this critique, we argue that their purported proof does not establish what it claims, and we give a counterexample to the result that their main theorem's purported proof centrally relies on.

\section{Preliminaries}\label{s:prelims}
 
Let $\naturalnumber = \{0, 1, 2, \ldots\}$ and let $\naturalnumberpositive = \{1, 2, 3, \ldots\}$. In this paper, we assume basic familiarity with Turing machines and computation, including time-bounded computation, using Turing machines.
Readers unfamiliar with those concepts may wish to consult the following  textbooks~\cite{hop-ull:b:automata,pap:b:complexity,du-ko:b:complexity}.
In their paper, Czerwinski~\cite{cze:t:p-neq-np-p-complete} assumes that the input alphabet of every Turing machine is $\{0,1\}$, which is unnatural, but in an attempt at doing as faithful of an analysis as possible, we also make that assumption. It's not too hard to see that the arguments we make still hold even if the input alphabet is an arbitrary finite set with at least two elements.
In this paper, we sometimes deviate slightly from Czerwinski's~\cite{cze:t:p-neq-np-p-complete} notations and typesetting in order to stay as close as possible to standard practices. For example, the language that we denote by $\up$ (defined in the next section) is denoted by $U_{\text{\bf P}}$ in~\cite{cze:t:p-neq-np-p-complete}.

\subsection{Classes and Oracle Machines}

As is standard, let $\p$ denote the class of decision problems accepted by deterministic polynomial-time Turing machines, and let $\np$ denote the class of decision problems accepted by nondeterministic polynomial-time Turing machines.
Additionally, we use the standard notation that $\dec$ is the class of recursive language and $\re$ is the class of recursively enumerable languages. 

We view an oracle as a language over a finite alphabet.
A deterministic oracle Turing machine (DOTM) $M$ with oracle $\calo$ (denoted $M^\calo$) is a type of deterministic Turing machine that has, in addition to the usual components of a Turing machine, an additional tape called the ``Query'' tape and three additional states called the ``Query'' state (which of course differs from the query tape despite having the same name), the ``Yes'' state, and the ``No'' state. The oracle machine has the special property that whenever it enters the Query state with $x$ on the Query tape, the next step of the machine is to either enter the Yes state if $x \in \calo$, or enter the No state if $x \not\in \calo$.
Nondeterministic oracle Turing machines (NOTMs) are defined analogously, and so are their time-bounded versions, i.e., deterministic polynomial-time oracle Turing machines (DPOTMs) and nondeterministic polynomial-time oracle Turing machines (NPOTMs).

For a given oracle $\calo$, $\p^\calo$ denotes the set of languages accepted by DPOTMs with oracle access to $\calo$. $\np^\calo$ and $\dec^\calo$ are defined analogously.
Given complexity classes $\calc, \cald$, the natural interpretation of $\calc^\cald = \bigcup_{A \in \cald} \calc^A$. Since it is easy to see that for every $A \in \p$, $\p^A = \p$ and $\np^A = \np$, it follows---and this is of course well-known---that $\p^\p = \p$ and $\np^\p = \np$. Thus it is clear that if the statement made in the title of Czerwinski's paper were validly established, Czerwinski would have unconditionally proved that $\p \neq \np$ \cite{cze:t:p-neq-np-p-complete}. 

\subsection{Reductions and Their Relationships}

In this paper, we consider several types of \emph{recursive} reductions, which we define below.

Given two sets $A$ and $B$, we say that $A$ Turing reduces to $B$ (denoted by $A \turing B$) exactly if there is a deterministic OTM $M$ such that $A = L(M^B)$ and $M$ halts on every input when $B$ is its oracle.

Given $f:\naturalnumber \rightarrow \naturalnumber$, the notation $A \boundedturing{f(n)} B$ means that $A$ Turing reduces to $B$ and the machine having oracle access to $B$ only makes up to $f(n)$ queries on each input on length $n$. 

Given two sets $A$ and $B$, we say that $A$ many-one reduces to $B$ (denoted by $A \manyone B$) exactly if there exists a computable function $f$ such that $(\forall x)[x \in A \iff f(x) \in B]$. It is well-known (and easy to see) that if $A \manyone B$, then $A \boundedturing{1} B$. The converse however is not true (see for example~\cite{hop-ull:b:automata}).

\section{Analysis of the Arguments\label{s:main}}

In their paper, Czerwinski~\cite{cze:t:p-neq-np-p-complete} define four sets that are used in their argument. Although they are not specific about it in their paper, we make the good-faith assumption that whenever the paper defines $M$ to be a Turing machine it is tacitly assumed that $M$ is a deterministic Turing machine (DTM)\@. And so, in what follows we assume that whenever we refer to $M$, we are speaking of a DTM\@. The sets in question are:

\begin{enumerate}
    \item $\up = \{ \langle M, x, 1^t\rangle \mid t\in \naturalnumberpositive \land M$ accepts $x$ within $t$ steps$\}$.
    \item $\unp = \{ \langle M, 1^n, 1^t\rangle \mid n \in \naturalnumber \land t \in \naturalnumberpositive \land (\exists x\in \{0,1\}^n)[M$ accepts $x$ within $t$ steps$]\}$.
    \item $\detp = \{ \langle M, x\rangle \mid M$ accepts $x\}$.
    \item $\detnp = \{ \langle M, 1^n\rangle \mid n \in \naturalnumber \land (\exists x\in \{0,1\}^n)[M$ accepts $x]\}$.
\end{enumerate}

The author makes the observation for every DTM $M$ and
\begin{enumerate}
\item every string $x$, $\langle M, x \rangle \in \detp \iff (\exists t > 0)[\langle M, x, 1^t\rangle \in \up]$.
\item every $n \in \naturalnumber$, $\langle M, 1^n\rangle \in \detnp \iff (\exists t > 0)[\langle M, 1^n, 1^t\rangle \in \unp]$.
\end{enumerate}

The $\p$-complete oracle alluded to in the title of Czerwinski's paper~\cite{cze:t:p-neq-np-p-complete} is the set $\up$, which is a well-known $\p$-complete problem.

We now examine the three purported proofs provided by~\cite{cze:t:p-neq-np-p-complete}, and find that they do not prove the paper's claim.

\subsection{Analysis of Czerwinski's Lemma~1}

We first state Czerwinski's Lemma~1 as it appears in the paper~\cite{cze:t:p-neq-np-p-complete}, except with some minor changes in the notation.
When examining Lemma~1 in Czerwinski's paper~\cite{cze:t:p-neq-np-p-complete}, it is unclear what specific claim is being established.

\begin{lemma}[{\cite[Lemma~1]{cze:t:p-neq-np-p-complete}}]\label{lemma}
If $\langle M, 1^n\rangle \in \detnp$, then it is undecidable, which $x \in \{0,1\}^n$ is in $L(M)$.
\end{lemma}

From this statement, both the ambiguity of whether $\langle M, 1^n \rangle$ is already known to be a member of $\detnp$ and the ambiguity of the term ``which'' create confusion in understanding Lemma~1's claim. Upon first investigation, Czerwinski's Lemma~1 seems to claim that there exists no computable function that outputs an $x\in\{0,1\}^n$ such that $x \in L(M)$ when given a TM $M$ and input string of fixed length $n$ such that $\langle M,1^n\rangle \in \detnp$. However, when Lemma~1 is later referenced in Czerwinski's Corollary 1, it can be interpreted as claiming that $\detnp$ is undecidable. In order to create as thorough a critique as possible, we will cover both of these interpretations of Czerwinski's Lemma~1. 

First, we will consider the claim that given $\langle M, 1^n\rangle\in\detnp$, there exists no computable function that outputs an $x \in L(M)$~\cite{cze:t:p-neq-np-p-complete}. Clearly if it is assumed that $\langle M, 1^n\rangle\in\detnp$, there must exist at least one $x\in\{0,1\}^n$ where $M$ on input $x$ halts and accepts. Therefore, we will show that a function for computing such an $x$ exists by defining an $f$ that iteratively simulates an increasing but finite number of steps of $M$ on each $w\in\{0,1\}^n$, and halts if any simulation accepts. It should be noted that this concept---known as ``dove-tailing''---can be applied to an infinite set of input strings, but since our input strings are bounded to have a length of $n$, we are restricted to only needing to test a finite number of strings from $\{0,1\}^n$. 

\begin{proposition}
\label{prop:dnp-computable-f}
    For every Turing machine $M$ and every string $1^n$ such that $\langle M, 1^n\rangle\in\detnp$, there exists a computable function $f$ that outputs a string of length $n$ in the language of $M$.
\end{proposition}
\begin{proof}
    Suppose $\langle M, 1^n\rangle \in \detnp$. 
    The approach we take here could be simplified greatly (indeed, if all $f$ has to do is output a string in $\{0,1\}^n \cap L(M)$ and the latter is nonempty (by assumption), then one of the constant functions that always returns a string in that set satisfies the proposition's requirements), but we take this approach as it generalizes in an obvious way that helps us establish the proof of Proposition~\ref{prop:dnp-undecidable}.
    
    We will proceed by describing a total machine $M'$ that computes the function $f$ described in the proposition statement on input $x$. $M'$ erases $x$ from its input tape and defines a constant $c \in\naturalnumberpositive$ (the actual choice of $c$ does not matter as long as it is a positive natural number).
    Next, $M'$ enters a main loop in which it iteratively simulates $M$ on the $\min(c, 2^n)$ lexicographically smallest strings of $\{0,1\}^n$ for exactly $c$ steps. During that step, if $M$ accepts on some string $w$, then $M'$ outputs $w$ and halts. If $M'$ finishes this iteration of the main loop without halting on any string, $c$ is increased by $1$ and the main loop is repeated. The function $f$ is clearly computable as all the steps in $M'$ are computable. Additionally, since simulating $M$ on each string is time-bound by $c$ steps, $M'$ will never run forever when simulating $M$ on a single string during its main loop.
    
    Clearly, for any $w\in \{0,1\}^n$, if $M$ on $w$ accepts, $M'$ on $w$ will return a string that is in $L(M)$. By the definition of $\detnp$, if $\langle M,1^n\rangle \in \detnp$ there must exist at least one $w\in\{0,1\}^n$ where $M$ accepts $w$. Therefore, since it is given that $\langle M, 1^n\rangle$ is in $\detnp$ and that $M'$ can simulate all possible values from $\{0,1\}^n$, this implies that at least one simulation of $M$ on a $w \in \{0,1\}^n$ in the inner loop of $M'$ will eventually halt and accept. Furthermore, because we are iteratively increasing $c$ after simulating on all $\min(c,2^n)$ lexicographically smallest strings in $\{0,1\}^n$ in each iteration, $M'$ will eventually halt for some $c$ and never run forever.
\end{proof}

Furthermore, we will consider the interpretation of Czerwinski's Lemma~1 that $\detnp$ is undecidable. In the paper's proof, it is claimed that it for a fixed $x \in \{0,1\}^n$ it is undecidable whether $x\in L(M)$, however, no explanation or proof of this vague statement is ever provided \cite{cze:t:p-neq-np-p-complete}. Additionally, the authors attempt to apply Rice's Theorem to show that for a fixed $x\in\{0,1\}^n$ it is undecidable if $x$ is a member of $\detnp$ (even if $x$ is the only element in $\detnp$) because ``it'' is a nontrivial property of $\detnp$~\cite{cze:t:p-neq-np-p-complete}. While this may be true, no further proof beyond this unclear sentence is used to back up these claims. 
Because of this, we will provide our own proof that membership in $\detnp$ is undecidable by a many-one reduction from a known undecidable language. And we will in the next section show that even if this interpretation of Lemma~1 is correct, their Corollary~1, which their main theorem centrally relies on, is false.

The halting problem, $\hp$, is a known undecidable problem~\cite{hop-mot-ull:b:intro-automata-languages-computation}. In this paper, we will define the halting problem as~\cite{hop-mot-ull:b:intro-automata-languages-computation}
$$\hp = \{\langle M,w\rangle \mid M \text{ halts on input }w\}.$$

Since the halting problem is known to be undecidable, we will prove that $\hp \manyone \detnp$ to show that $\detnp$ is also undecidable. 

\begin{proposition}\label{prop:dnp-undecidable}
    The language $\detnp$ is undecidable.
\end{proposition}
\begin{proof}
    To show that $\hp$ many-one reduces to $\detnp$, we will give a computable function $f$ such that $(\forall x)[x \in \hp \iff f(x) \in \detnp]$. We will define $f$ as follows.
    
    First, $f$ will attempt to decode $x$ into $\langle M,w \rangle$, where $M$ is a TM and $w$ is a string. If $x$ cannot be decoded, $f$ will return a fixed string $z \not\in\detnp$. Otherwise, $f$ will construct a TM $M'$ that takes an input $y$, and checks if $y$ is equal to $\epsilon$. If $y$ is not equal to $\epsilon$, $M'$ will immediately reject. However, if $y=\epsilon$, $M'$ immersively simulates $M$ on $w$. If $f$ did not return $z$, $f$ will return the pair $\langle M',\epsilon\rangle$.

    If $x$ cannot be decoded into a TM $M$ and string $w$, $f(x)$ will output $z\notin \detnp$. Conversely, if $x$ is a valid encoding, $f$ will return a TM $M'$ and the string $\epsilon$. Therefore, since $M'$ can be constructed for all $x$, $f(x)$ is a computable function.
    
    For all syntactically valid encoding of $x$, simulating $f(x)$'s resulting TM $M'$ on its resulting string $\epsilon$ will always cause $M'$ to simulate $x$'s encoded $M$ on $w$. Therefore, if $x \in \hp$, $f(x)$ will create a TM $M'$ that halts on input $\epsilon$. Because of this, $\epsilon \in L(M')$, so $f(x) \in \detnp$ by $\detnp$'s definition. However, if $x \notin \hp$, $f(x)$ will create a TM $M'$ that will never halt on input $\epsilon$. Since this $M'$ will run forever on $\epsilon$, it is clear that $\epsilon \notin L(M')$, which means that $(\forall x)[x \notin \hp \implies f(x) \notin \detnp]$. Since we have proven that $f$ is a computable function and that $f$ satisfies the biconditional for all $x$, $\hp \manyone \detnp$. 
\end{proof}

\subsection{Analysis of Czerwinski's Corollary~1}
\label{subsec:analysis-cze-cor-1}

In Corollary~1, Czerwinski claims that for arbitrary TM $M$ and $n \in \naturalnumber$, testing whether $\pair{M, 1^n} \in \detnp$ requires $\Theta(2^n)$ queries to the oracle $\detp$ in the worst case~\cite{cze:t:p-neq-np-p-complete}, i.e.\ for any function $h$ such that $\detnp \boundedturing{h} \detp$, $h$ must be in $\Omega(2^n)$. Their purported proof of Corollary~1 lacks specifics and only concludes that the ``black-box complexity''~\cite{cze:t:p-neq-np-p-complete} (an unclear notion that they do not define in the paper; and we do not attempt to define it) is exponential: It is not apparent how Czerwinski arrived at the $\Omega(2^n)$ oracle queries claimed in their corollary. We give our understanding of their purported proof, but only to bring clarity to confusing aspects of their proof: We later show that computing membership in $\detnp$ can be done with a single query to the oracle $\detp$ and that Czerwinski's Corollary~1 is false.

Czerwinski's purported proof of Corollary~1 first claims that for an arbitrary TM $M$ and $n \in \naturalnumber$, testing whether an input $x$ of length $n \in \naturalnumber$ is in $L(M)$ is undecidable~\cite{cze:t:p-neq-np-p-complete}. It is ambiguous as to what set they are referring to: Either they are referring to $\detnp$ and simply restating that it is undecidable, or they are referring to a new set that, as described, is analogous to $\detp$ but with a bound on the length of $x$. Czerwinski then observes that the characteristic function of $L(M)$ is computable using $\detp$ as an oracle~\cite{cze:t:p-neq-np-p-complete}. They do not detail how to use the characteristic function to arrive at an exponential lower bound for the ``black-box complexity''~\cite{cze:t:p-neq-np-p-complete} of testing membership in $\detnp$ using $\detp$ as an oracle. We assume that Czerwinski is basing their purported proof on the idea that if computing the characteristic function of $\detnp$ requires exponential ``black-box complexity''~\cite{cze:t:p-neq-np-p-complete}, then computing membership in $\detnp$ requires exponential time. Towards that end, we further assume that Czerwinski is using an unstated proof by contradiction, where a faster-than-exponential, i.e. polynomial ``black-box complexity''~\cite{cze:t:p-neq-np-p-complete} for computing the characteristic function of $\detnp$ would make some known undecidable problem, such as the halting problem, decidable.

We now show that Czerwinski's Corollary~1 is false. Specifically, we use the notions of $\re$-completeness, many-one reductions, and Turing reductions to prove that membership in $\detnp$ can be computed with a single query to the oracle $\detp$. As a preliminary to our proof of Proposition~\ref{prop:dnp-btr-dp}, a set $A$ is $\re$-complete exactly if 1.~$A \in \re$, and 2.~$(\forall L \in \re)[L \manyone A]$.

\begin{proposition}\label{prop:dnp-btr-dp}
    $\detnp\boundedturing{1}\detp$.
\end{proposition}
\begin{proof}
    It can be trivially shown via a technique similar to one in the proof of Proposition~\ref{prop:dnp-computable-f} that $\detnp \in \re$, i.e.\ that it is recursively enumerable. The oracle set $\detp$ is known to be $\re$-complete (see~\cite{hop-ull:b:automata}, which itself credits~\cite{tur:j:computable-numbers}). Therefore, it must be the case that $\detnp \manyone \detp$, which, by the properties of many-one reductions, implies that $\detnp \boundedturing{1} \detp$.\footnote{In fact, though we do not draw on it in this paper, our proofs of Propositions~\ref{prop:dnp-undecidable} and~\ref{prop:dnp-btr-dp} tacitly imply that $\detnp$ is $\re$-complete as the latter proof establishes that $\detnp\in\re$ and the former proof shows that $\hp\manyone\detnp$, and since recursive many-one reductions are transitive, it follows that for every $L\in\re$, $L \manyone \detnp$. In light of this completeness result, we could also use the well-known fact that $\re$-complete sets are recursive isomorphic to complete the proof of Proposition~\ref{prop:dnp-btr-dp}, but that is an overly-complicated approach.}
\end{proof}

From Proposition~\ref{prop:dnp-btr-dp}, it is obvious that $\detnp$ can be computed with a single query to the oracle $\detp$. Furthermore, this is optimal with respect to the number of queries within the Turing reduction framework, as $\detnp \boundedturing{0} \detp$ is equivalent to saying $\detnp$ is recursive, which is clearly false.

A more informal description of testing membership in $\detnp$ with a single oracle query is as follows: For any instance of the decision problem version of $\detnp$ (testing whether an instance is in $\detnp$), we can many-one reduce it to an instance of the decision problem version of $\detp$, the oracle set, as $\detnp \in \re$ and $\detp$ is $\re$-complete. Because the reduced instance is in $\detp$ if and only if the initial instance is in $\detnp$, we can test whether the initial instance is in $\detnp$ with a single query to $\detp$ on the reduced instance. Thus Czerwinski's claim in Corollary~1 that testing membership in $\detnp$ requires $\Theta(2^n)$ accesses to the oracle $\detp$~\cite{cze:t:p-neq-np-p-complete} is false.

\subsection{Analysis of Czerwinski's Theorem~1}

Czerwinski's Theorem~1 states that for any arbitrary TM $M$, we need $\Theta(2^n)$ accesses to the oracle $\up$ to decide whether $\langle M,1^n,1^t\rangle \in \unp$, where $n,t\in \naturalnumberpositive$~\cite{cze:t:p-neq-np-p-complete}.

Czerwinski claims that for any TM $M$ and natural number $n$, there is a time $t' \in \naturalnumber$ such that there exists a string $x$ of length $n$ that $M$ accepts if and only if there exists a string $x$ of length $n$ that $M$ accepts in at most $t'$ steps~\cite{cze:t:p-neq-np-p-complete}. Hence, there is a biconditional relationship between $\unp$ and $\detnp$, and also between $\up$ and $\detp$. Czerwinski then uses these relationships to infer that the complexity of accepting $\unp$ relative to oracle $\up$ is the same as the complexity of accepting $\detnp$ relative to oracle $\detp$~\cite{cze:t:p-neq-np-p-complete}. Czerwinski concludes the purported proof by attempting to show that these time complexities are optimal.

However, as mentioned in Section~\ref{subsec:analysis-cze-cor-1}, we only need one call to oracle $\detp$ so we might not need $\Theta(2^n)$ calls to the oracle $\up$ when accepting $\unp$. And so, one could possibly argue that the time complexity of accepting $\unp$ relative to oracle $\up$ is not the same as that of accepting $\detnp$ relative to oracle $\detp$.

In conclusion, as Czerwinski's Theorem~1 relies mostly on Corollary~1 and by Section~\ref{subsec:analysis-cze-cor-1}, this proof seems not to establish the result. Moreover, the proof of the similar time complexity between converting $\unp$ and $\detnp$, and between $\up$ and $\detp$ is not clear enough. This leads us to doubt the veracity of that statement.

\section{Conclusion}\label{s:conclusion}

Through this paper, we have established that the crux of Czerwinski's paper~\cite{cze:t:p-neq-np-p-complete}, i.e., its Corollary~1, is false by showing that $\detnp \boundedturing{1} \detp$. Furthermore, we find another important flaw in the purported proof of the paper's Theorem~1, which asserts that if $\detnp \not\in \p^{\detp}$, then $\unp \not\in \p^{\up}$. However, the paper's argument confuses its own notion of ``black box complexity'' with the notion of ``time complexity.'' Indeed, in our Proposition~\ref{prop:dnp-btr-dp}, we establish that one oracle query suffices and yet, this sheds no light on the number of steps taken by the ``base machine'' (i.e., the oracle machine). Furthermore, the paper~\cite{cze:t:p-neq-np-p-complete} claims to separate $\p^\p$ and $\np^\p$, which would be equivalent to $\p \neq \np$. Thus the main theorem of the paper is false and does not show that $\p \neq \np$.

\paragraph{Acknowledgements}

We would like to thank
Yumeng He,
Lane A. Hemaspaandra,
Matan Kotler-Berkowitz, and
Zeyu Nie
for their helpful comments on prior drafts.
The authors are responsible for any remaining errors.

\bibliographystyle{alpha}
\bibliography{gry-reu,local_refs}

\end{document}